\documentclass[prx,twocolumn,superscriptaddress,longbibliography,notitlepage]{revtex4-2}
\usepackage[latin1]{inputenc}
\usepackage{graphicx}
\usepackage{amsmath}
\usepackage{amsthm}
\usepackage{bm}
\usepackage{layout}
\usepackage{float}
\usepackage{amsfonts}
\usepackage{amssymb}%
\usepackage[margin=0.75in]{geometry}
\usepackage{color}
\usepackage{soul}
\usepackage{mathtools}
\usepackage[colorlinks=true,citecolor=blue]{hyperref}
\theoremstyle{definition}

\newtheorem{lemma}{Lemma}

\usepackage[capitalise,compress]{cleveref}
\newcommand{\avg}[1]{\left \langle #1 \right\rangle}

\newcommand{\Tr}{\mathrm{Tr}}

\newcommand{\abs}[1]{\left | #1 \right|}
\renewcommand{\epsilon}{\varepsilon}
\renewcommand{\O}[1]{\mathcal{O}\left(#1\right)}

\newcommand{\norm}[1]{\left\|#1\right\|}

\newcounter{para}

\makeatletter
\newcommand*\bigcdot{\mathpalette\bigcdot@{.5}}
\newcommand*\bigcdot@[2]{\mathbin{\vcenter{\hbox{\scalebox{#2}{$\m@th#1\bullet$}}}}}

\makeatother

\usepackage{array}
\newcolumntype{L}{>{$}l<{$}} 
\newcolumntype{C}{>{$}c<{$}} 
\newcolumntype{R}{>{$}r<{$}} 

\usepackage{xcolor}

\usepackage{mathtools}

\usepackage{xr}
\makeatletter
\newcommand*{\addFileDependency}[1]{
  \typeout{(#1)}
  \@addtofilelist{#1}
  \IfFileExists{#1}{}{\typeout{No file #1.}}
}
\makeatother

\usepackage{tikz}

\usepackage{mdframed}
\newmdenv[topline=false,rightline=false,bottomline=false,linewidth=2pt,linecolor=white!60!black,]{leftborder}

\usepackage{qcircuit}

\newcommand{\supp}[1]{\text{supp}(#1)}

\makeatletter
\newcommand{\subalign}[1]{%
  \vcenter{%
    \Let@ \restore@math@cr \default@tag
    \baselineskip\fontdimen10 \scriptfont\tw@
    \advance\baselineskip\fontdimen12 \scriptfont\tw@
    \lineskip\thr@@\fontdimen8 \scriptfont\thr@@
    \lineskiplimit\lineskip
    \ialign{\hfil$\m@th\scriptstyle##$&$\m@th\scriptstyle{}##$\hfil\crcr
      #1\crcr
    }%
  }%
}
\makeatother

\usepackage{physics}

\newcommand{\NC}{\mathcal{N}}
\newcommand{\NU}{\tilde{\mathcal{U}}}
\newcommand{\UC}{\mathcal{U}}

\newcommand{\PEC}{\text{PEC}}
\newcommand{\LoPEC}{\text{LoPEC}}
\newcommand{\Var}{\text{Var}}

\newcommand{\Ham}{\text{Ham}}
\newcommand{\p}{p}

\begin{document}
\title{Locality and Error Mitigation of Quantum Circuits}
\date{\today}  
\author{Minh C. Tran}
\affiliation{IBM~Quantum,~IBM~T.J.~Watson~Research~Center,~Yorktown~Heights,~NY~10598,~USA}
\author{Kunal Sharma}
\affiliation{IBM~Quantum,~IBM~T.J.~Watson~Research~Center,~Yorktown~Heights,~NY~10598,~USA}%
\author{Kristan Temme}
\affiliation{IBM~Quantum,~IBM~T.J.~Watson~Research~Center,~Yorktown~Heights,~NY~10598,~USA}%

\begin{abstract}
In this work, we study and improve two leading error mitigation techniques, namely Probabilistic Error Cancellation (PEC) and Zero-Noise Extrapolation (ZNE), for estimating the expectation value of local observables. For PEC, we introduce a new estimator that takes into account the light cone of the unitary circuit with respect to a target local observable.
Given a fixed error tolerance, the sampling overhead for the new estimator can be several orders of magnitude smaller than the standard PEC estimators. 
For ZNE, we also use light-cone arguments to establish an error bound that closely captures the behavior of the bias that remains after extrapolation.
\end{abstract}
\maketitle

\section{Introduction} Noisy quantum computers are scaling beyond the point their classical counterparts can efficiently simulate. A central question is whether, without error correction, noisy quantum devices can provide practical advantages over classical methods \cite{m:bravyiFutureQuantumComputing2022}. Most near-term algorithms involve the estimation of the expectation value of an observable after a shallow-depth circuit~\cite{m:corcolesChallengesOpportunitiesNearTerm2020a} and, to achieve quantum advantages, this estimation should be more accurate than those produced by classical algorithms~\cite{m:temmeErrorMitigationShortdepth2017a}.

Besides algorithmic errors, experimental inaccuracies degrade the desired output expectation values. The primary goal of error mitigation is to mitigate this effect and produce a better estimation of the expectation value, at the expense of increased sampling overhead and classical post-processing.
Based on this idea, numerous mitigation techniques, such as Probabilistic Error Cancellation (PEC) \cite{m:temmeErrorMitigationShortdepth2017a,m:liEfficientVariationalQuantum2017}, Zero-Noise Extrapolation (ZNE) \cite{m:temmeErrorMitigationShortdepth2017a}, and application-specific mitigation strategies employing symmetries and post-selection techniques have been introduced and demonstrated in recent years \cite{m:mccleanTheoryVariationalHybrid2016,
m:mccleanHybridQuantumclassicalHierarchy2017,
m:obrienErrorMitigationVerified2021,
m:piveteauErrorMitigationUniversal2021a,
m:kandalaErrorMitigationExtends2019,
m:endoPracticalQuantumError2018,
m:bonet-monroigLowcostErrorMitigation2018,
m:sagastizabalExperimentalErrorMitigation2019,
m:mcardleErrorMitigatedDigitalQuantum2019,
m:suErrorMitigationNearterm2021,
m:czarnikErrorMitigationClifford2021,
m:hugginsVirtualDistillationQuantum2021b,
m:koczorExponentialErrorSuppression2021,
m:piveteauErrorMitigationUniversal2021a,
m:cerezoVariationalQuantumState2022,
m:cincioMachineLearningNoiseResilient2021,
m:funckeMeasurementErrorMitigation2022,
m:mccleanHybridQuantumclassicalHierarchy2017,
m:wangCanErrorMitigation2021,
mm:bergSingleshotErrorMitigation2022,
m:takagiFundamentalLimitsQuantum2022,
m:takagiUniversalSamplingLower2022}. 
These techniques exploit only the structure of the noisy circuits and are applicable regardless of the observables we would like to measure.

In this work, we analyze the performance of two leading error mitigation strategies, namely PEC and ZNE, and show that they can be drastically improved when specialized to local observables. 
For PEC, we introduce an efficient estimator for the expectation value of an observable that takes into account its ``light cone''~\cite{m:tranHierarchyLinearLight2020a,m:childsTheoryTrotterError2021}, reducing the sampling overhead of PEC by several orders of magnitude. Additionally, our estimator does not require tailoring the experimental procedure to the observables and, thus, can be retroactively applied to improve estimates from past PEC experiments.

While PEC estimators are unbiased and sampling overhead is the primary bottleneck of PEC, the main challenge in ZNE is to constrain the bias that remains after applying the Richardson extrapolation~\cite{m:temmeErrorMitigationShortdepth2017a}.
Here, we prove a bound on the bias in ZNE using light-cone arguments, showing ZNE performs significantly better when applied to local observables.
In particular, our bound increases with the number of gates inside the light cones and closely capture the correct behavior of the bias.

\medskip 

\section{Setup} We consider an ideal circuit consisting of $d$ unitary channels $\mathcal U = \mathcal U_d \cdots \mathcal U_1$. Each unitary channel may be a single quantum gate or a layer consisting of multiple gates. Let $\tilde{\mathcal U} = \tilde{\mathcal U}_d\cdots \tilde{\mathcal U}_1$ be the corresponding noisy circuit, where each noisy channel $\tilde{\mathcal U}_i$ is the composition of the ideal channel $\mathcal U_i$ with a noise channel $\mathcal N_i$:
\begin{align} 
	\tilde {\mathcal{U}}_i =  \mathcal U_i \mathcal N_i. 
\end{align}
The noise channels $\mathcal N_i$ can be assumed to be Pauli channels when the idealized gates $\mathcal U_i$ are Clifford gates. Generalized noise channels can always be transformed into Pauli channels via Pauli twirling ~\cite{m:bennettPurificationNoisyEntanglement1996,m:bennettPurificationNoisyEntanglement1996,m:knillFaultTolerantPostselectedQuantum2004,m:kernQuantumErrorCorrection2005,m:gellerEfficientErrorModels2013,m:wallmanNoiseTailoringScalable2016,m:bergProbabilisticErrorCancellation2022a}. A key prerequisite in many state-of-the-art error mitigation techniques is the ability to learn the noise channels $\mathcal N_i$. Fully characterizing a Pauli channel on $n$ qubits generally requires resources scaling exponentially with $n$, but it can often be made efficient by considering the structure of $\mathcal N_i$ on a particular device. To learn $\mathcal N_i$, we make the assumption that it can be written in Pauli-Lindblad form: $\mathcal N_i = e^{\mathcal L_i}$, where the generator
\begin{align} 
	\mathcal L_i(\rho)  = \sum_{j = 1}^J \lambda_{i,j} (P_{i,j} \rho P_{i,j}^\dag - \rho)
	= \sum_{j = 1}^J \mathcal L_{i,j}(\rho)\label{eq:generator}
\end{align}
for some $J\in \mathbb N$, $\mathcal L_{i,j}(\rho)\equiv \lambda_{i,j} (P_{i,j} \rho P_{i,j}^\dag - \rho)$, and $P_{i,j}$ are Pauli strings on $n$ qubits. The task of learning a general Pauli-Lindblad channel $\mathcal N_i$ reduces to learning possibly $\O{4^n}
$ parameters $\lambda_{i,j}$.
However, if $\mathcal U_i$ is a constant-depth layer, it is reasonable to assume that only low-weight Pauli strings contribute to the generator $\mathcal L_i$. 
It is then sufficient to learn the corresponding polynomially many parameters.
Throughout this paper, we write the error channels in this Pauli-Lindblad form and assume they can be efficiently learned.
Other than that, we do not make any locality assumption about the noise model.
So, for example, qubits that are not connected by $\UC_i$ can still have correlated Pauli noise.

\begin{figure*}[t]
\centering
\includegraphics[width=0.7\textwidth]{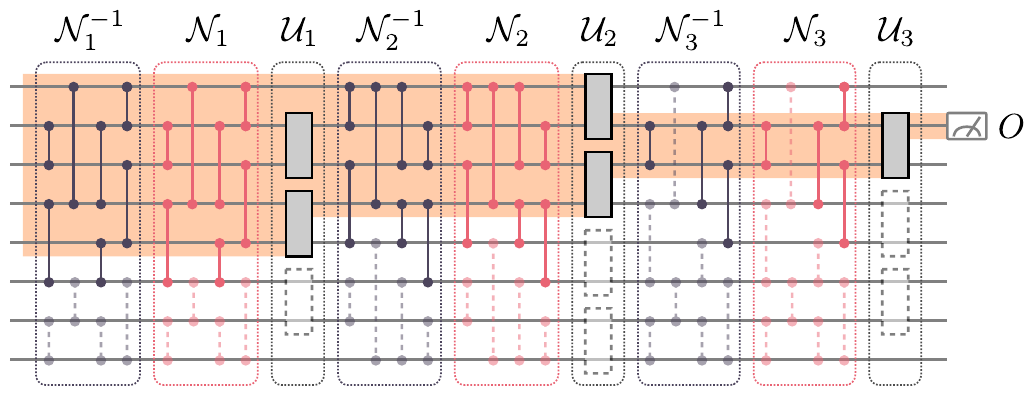}
\caption{An illustration of the light cone of the observable $O$ initially supported on a single site. 
The ideal circuit $\UC_3\UC_2\UC_1$ consists of three layers, each consists of gates (rectangular boxes) that may or may not act nontrivially on $O$ (filled and empty boxes, respectively).
The orange shaded area indicates the qubits inside the light cone $\mu_i$ defined in \cref{eq:light-cone-def}.
The noisy version of the circuit consists of additional noise channels $\NC_1,\NC_2,$ and $\NC_3$, which can be further decomposed into a series of Pauli channels in \cref{eq:expansion-Lambda} (vertical lines). 
PEC cancels errors in the circuit by probabilistically applying the inverse $\NC_i^{-1}$ at each layer.
We reduce the sampling overhead of PEC by removing the contributions from error channels that lie outside the light cone (dashed vertical lines) in the construction of the local estimator [\cref{eq:local-PEC-estimator}].
}
\label{fig:lightcone-demo}
\end{figure*}

We note that, since two Pauli strings either commute or anti-commute, the generators $\mathcal L_{i,j}$ in \cref{eq:generator} mutually commute.
The noise channel therefore factorizes into $J$ channels: $\mathcal N_i = \prod_{j = 1}^J \mathcal N_{i,j}$, where
\begin{align} 
	\NC_{i,j}(\rho) = e^{\mathcal L_{i,j}}(\rho) = (1-p_{i,j})\rho + \p_{i,j} P_{i,j} \rho P_{i,j}^\dag, \label{eq:expansion-Lambda}
\end{align}
and $\p_{i,j} = (1-e^{-2\lambda_{i,j}})/2\in [0,1/2)$ is the probability of the error $P_{i,j}$.
Eq.~\eqref{eq:expansion-Lambda} follows from the Taylor expansion of $\exp(\mathcal L_k)$ and the fact that $\mathcal L_k^2 = -2\mathcal L_k$.  
In terms of $\NC_{i,j}$, the full noisy circuit is given by 
\begin{align} 
	\tilde{\UC} =\prod_{i = 1}^d \NU_i = \prod_{i = 1}^d \UC_i \prod_{k = 1}^{J}\NC_{i,j}, \label{eq:full-noisy}
\end{align}
where $\prod_{i} \mathcal A_i = \cdots \mathcal A_2 \mathcal A_1$ is an ordered product.

Upon applying $\tilde{\mathcal U}$ on an initial state $\rho$, we measure in the computational basis $z\in \{0,1\}^{n}$ to obtain $m$ samples $\{z_1,\dots,z_m\}$.
We use the samples to estimate the expectation value of an observable $O = \sum_z o_z \ketbra{z}$:
	$\Tr[\NU(\rho)O] = \Tr[\rho\ \NU^\dag(O)] \approx \frac{1}{m} \sum_{j=1}^mo_{z_j},$
where $\NU^\dag$ is the adjoint map of $\NU$.
We assume that $\norm{O} = 1$, where $\norm{\cdot}$ is the spectral norm, throughout this paper. 
Our main results apply to local observables---those supported on only a few sites, such as two-point correlators---and linear combinations of local observables.

\medskip

\section{PEC for local observables}The presence of noise deviates the expectation value of $O$ from its ideal value.
PEC mitigates this effect by additionally implementing the inverse of each noise channel:
\begin{align} 
	{\NU}_{\PEC} &\equiv \prod_{i = 1}^d \NU_i \left(\prod_{j = 1}^{J}\NC^{-1}_{i,j}\right)\nonumber\\
	&=\prod_{i = 1}^d \UC_i \left(\prod_{k' = 1}^{J}\NC_{i,k'}\right) \left(\prod_{j = 1}^{J}\NC^{-1}_{i,j}\right) = \prod_{i = 1}^d \UC_i.  \label{eq:PEC-circuit}
\end{align}
Compared to the bare noisy circuit in \cref{eq:full-noisy}, the addition of $\NC^{-1}_{i,j}$ exactly negates $\NC_{i,j}$ for all $i,j$, returning the ideal circuit.
One can exactly derive the inverse maps $\NC^{-1}_{i,j}$ upon learning $\mathcal N_{i,j}$:
\begin{align} 
	 \NC^{-1}_{i,j}(\rho)  = \gamma_{i,j}\left[(1-\p_{i,j}) \rho - \p_{i,j} P_{i,j}^{\dag} \rho P_{i,j}\right],
\end{align}
where $\gamma_{i,j} = e^{2\lambda_{i,j}} \geq 1$.
To probabilistically implement $\NC^{-1}_{i,j}(\rho)$, in each run of the circuit, we additionally apply $P_{i,j}$ before each $\NU_i$ with probability $\p_{i,j}$ and appropriately rescale the measurement results to account for the normalization factor $\gamma_{i,j}$. 
Explicitly, let $\sigma_{i,j} = \pm 1 $ indicate whether $P_{i,j}$ was applied $(-1)$ or not $(1)$ and  $\bm \sigma \in \{\pm 1\}^{dJ}$ be a vector containing all $\sigma_{i,j}$.
Ref.~\cite{m:temmeErrorMitigationShortdepth2017a} shows that 
\begin{align} 
	\hat o_z^{\PEC}(\bm \sigma) = o_z(\bm \sigma) \prod_{i,j} \gamma_{i,j} \sigma_{i,j}  \label{eq:o-PEC-estimator}
\end{align}
is an unbiased estimator for the ideal expectation value:
\begin{align} 
	\mathbb E_{\bm \sigma} \mathbb E_z [\hat o_z^{\PEC}(\bm \sigma)] =  \avg{\UC^\dag(O)} \label{eq:PEC-convergent}
\end{align}
where $\avg{\cdot}$ is the expectation value with respect to $\rho$.
For readability, we may drop the implicit dependence on $\bm \sigma$ and simply write $o_z$ and $\hat o_z^{\PEC}$ in the rest of the paper.
\cref{eq:o-PEC-estimator} indicates that the variance of $\hat o_z^{\PEC}$,
\begin{align} 
	\Var[\hat o_z^{\PEC}] = \O{\prod_{i,j} \gamma_{i,j}^2} 
	= \O{e^{4\sum_{i,j}\lambda_{i,j}}} ,\label{eq:var-oz-PEC}
\end{align}
is exponential in the total error rate $\lambda \equiv \sum_{i,j}\lambda_{i,j}$.
For a generic depth-$d$ circuit on $n$ qubits, $\lambda = \O{nd}$, implying that the sampling overhead---the number of shots $m$ it takes for $\hat o_z^{\PEC}$ to converge to the true value---grows exponentially both with the system size and with the depth of the circuit.

Given a local observable $O$, our first result is a construction of an efficient estimator $o^{\LoPEC}_z$ whose variance depends only on the ``light cone'' of $\mathcal U$ with respect to $O$ and can be much smaller than that of $\hat o_z^{\PEC}$.
Our construction is based on an observation that the maps $\NC_{i,j}$ and their inverse should not affect the expectation of $O$ if they are supported entirely outside the light cone.
Therefore, $\p_{i,j}$ and the corresponding $\lambda_{i,j}$ should not contribute to the uncertainty in estimating $\avg{O}$.

For a concrete analysis, we first define the light cone of an observable $O$.
We consider the Heisenberg picture, where an adjoint map $\NU_\PEC^\dag$ is applied on $O$:
\begin{align} 
	 {\NU}_{\PEC}^\dag =  \prod_{i = d}^1 \left(\prod_{j = 1}^{J}\NC^{-1}_{i,j}\right) \NU_i^\dag. 
\end{align}
We note that both $\NC_{i,j}$ and their inverse are self-adjoint and mutually commute.
In this picture, $\NU_d^\dag$ is the first channel to be applied on $O$.
Let us assume that $O$ is a local Pauli operator. 
Let  $\supp{O}$ denote the support of $O$ and
\begin{align} 
	\mu_i \supseteq \supp{\UC_i^\dag \UC_{i+1}^\dag \cdots \UC_{d}^\dag(O)} \label{eq:light-cone-def}
\end{align}
be a set of qubits that contains the support of $O$ after evolving under the last ideal $d-i+1$ unitary channels.
We require that $\mu_i$ are efficiently classically computable for all $i$. 
There may be multiple choices of $\mu_i$, but for the reasons we discuss below, it is ideal to pick the smallest sets that satisfy these requirements. 
The collection of $\mu_i$ for $i = 1,\dots,d$ forms what we call the light cone of $O$ under $\UC$. 
We say that a map $\mathcal N_{i,j}$ is outside the light cone if its support is distinct from $\mu_i$, i.e. $\supp{P_{i,j}}\cap \mu_i = \varnothing$, and is inside otherwise. 
We note that Pauli channels $\NC_{i,j}$ do not increase the support of $O$, allowing the light cone $\mu_i$ to be characterized entirely using the ideal circuit.

Let $\mu$ be the set of all $(i,j)$ such that $\mathcal N_{i,j}$ are inside the light cone.
We claim that 
\begin{align} 
	\hat o_z^{\LoPEC}(\bm \sigma) \equiv o_z (\bm \sigma) \prod_{(i,j)\in \mu} \gamma_{i,j} \sigma_{i,j}   \label{eq:local-PEC-estimator}
\end{align}
is also an unbiased estimator for the ideal expectation value:
\begin{align} 
	\mathbb E_{\bm \sigma}\mathbb E_{z}[\hat o_z^{\LoPEC}] = \avg{\UC^\dag(O)}. \label{eq:LoPEC-convergence}
\end{align}
We emphasize that $o_z$ is the same outcomes after applying the standard PEC circuit in \cref{eq:PEC-circuit}.
But, the local PEC estimator $\hat o_z^{\LoPEC}$ has a simple, intuitive feature: its variance
\begin{align} 
	\Var[\hat o_z^{\LoPEC}] = \O {e^{4\sum_{(i,j)\in \mu}\lambda_{i,j}}} \label{eq:LoPEC-convergent}
\end{align}
only involves $\lambda_{i,j}$ that correspond to noise channels inside the light cone and may be much smaller than that of $\hat o_z^{\PEC}$, as defined in \cref{eq:var-oz-PEC}.

We provide detail proof for our claims in \cref{sec:lopec} and illustrate its idea using a simple example.
Consider an ideal circuit consisting of a single layer $\UC$ ($d = 1$) and a single Pauli noise channel $\NC $ ($J = 1$):
\begin{align} 
	\NC(\rho) = (1-p)\rho + p P \rho P^\dag, 
\end{align}
where $P$ is a Pauli string supported entirely outside the support of $\UC^\dag( O)$.
The implementation of PEC would involve applying either $\NU = \UC \NC$ with probability $1-p$ or $\NU \mathcal P$, where $\mathcal P(\rho) = P\rho P^\dag $, with probability $p$.
On average, the standard PEC estimator in \cref{eq:o-PEC-estimator} returns the ideal expectation value
\begin{align} 
	&\gamma\left[(1-p)\avg{ \NU^\dag(O)} - p \avg{\mathcal P\NU^\dag(O)}\right] \nonumber\\
	&=  \avg{\NC^{-1}\NU^\dag(O)} =  \avg{\UC^\dag(O)},
\end{align}
with a variance $\O{\gamma^2}$.
In contrast, because $\NC$ is outside the light cone, the local estimator in \cref{eq:local-PEC-estimator} gives
\begin{align} 
	(1-p)\avg{ \NU^\dag(O)} + p\avg{ \mathcal P\NU^\dag(O)} = \avg {\NC\NU^\dag(O)} 
\end{align}
on average.
Coupling with the fact that $\NC$ acts trivially on $\UC^\dag(O)$, we have
$	\avg {\NC\NU^\dag(O)} =\avg{\NC^2\UC^\dag(O)} = \avg {\UC^\dag(O)}, $
which is also the ideal expectation value.
However, the variance of this local estimator is $\O{1}$ instead of $\O{\gamma^2}$.
Generalizing this argument to $d,J\geq 1$ proves our claim in \cref{eq:LoPEC-convergent}.

To demonstrate the performance of the local PEC estimator, we numerically simulate the expectation value of a local observable after a noisy depth-$d$ circuit on $n$ qubits.
The $n$ qubits are arranged on a two-dimensional heavy hex lattice.
Each layer of the circuit contains approximately $n/2$ CNOTs between distinct pairs of nearest-neighboring qubits.
At each layer, we generate $J = 10n$ random noise channels such that the error rate per CNOT is approximately $4\times 10^{-3}$.
See \cref{sec:numerics-specifications} for the exact specifications of the random noisy~circuit.
 
In \cref{fig:pec-numerics}(a), we plot a histogram of the mean values of the standard PEC estimator $\hat o_z^{\PEC}$ (gray) and the local PEC estimator $\hat o_z^{\LoPEC}$ (orange) after $m = 10^6$ samples. We obtain the histogram by generating $10^4$ independent sets of $m$ samples from the same noisy circuit.
Here, $n = 65, d = 35$, and $O = Z_1$ is the Pauli Z on the first qubit. 
The inset plots the sites inside the light cone from $O$ (orange dots) at different layers. 
We choose the circuit to be deep enough so that the light cone may cover the entire system.
The figure demonstrates that both $\hat o_z^{\PEC}$ and $\hat o_z^{\LoPEC}$ are unbiased estimators of the ideal value, but the latter has a smaller variance, thus requiring much fewer samples to converge.

In \cref{fig:pec-numerics}(b), we fix the depth of the circuits ($d = 10$) and plot the sampling overhead for estimating several observables---namely $Z_1, Z_1Z_{10},$ and $Z_1Z_4Z_{10}$---given circuits of different $n$.
Here, the sampling overhead is the variance of the estimator divided by $\epsilon^2$, where $\epsilon = 0.01$ is the desired precision.
While the sampling overhead using the standard PEC estimator grows with the system size, the sampling overhead of the local PEC estimator saturates at large $n$.

\section{ZNE for local observables}The main idea of ZNE is to obtain expectation values at several different error rates and, from these values, extrapolate to the zero noise limit. 
In contrast to PEC, ZNE is an biased estimator, and the error cannot be completely suppressed by taking more samples. 
In addition, the ZNE error cannot be inferred from the measurement statistics and one must rely on theoretical error bounds to guarantee the outcome is close to the ground truth. 
Here, we establish an error bound for ZNE that takes into account the light cone of the target observable.

In ZNE, we obtain expectation values at different error rates by amplifying the noise channels.
Let $\mathcal N_{i,j}^{(g)}$ denote the amplified version of $\mathcal N_{i,j}$, constructed from $\mathcal N_{i,j}$ by replacing $\p_{i,j}$ with $g\p_{i,j}$ for some $g\geq1$.
Note that instead of amplifying $\lambda_{i,j}$ as in earlier works~\cite{m:temmeErrorMitigationShortdepth2017a,m:bergProbabilisticErrorCancellation2022a}, we amplify the error probability $p_{i,j}$ for simplicity.
To implement such amplification with a gain factor $g\geq 1$:
\begin{align} 
	{\NU}^{(g)} &\equiv \prod_{i = 1}^d \UC_i \prod_{j = 1}^{J}\NC^{(g)}_{i,j}
	=\prod_{i = 1}^d \NU_i \prod_{j = 1}^{J}\NC^{(\delta g_{i,j})}_{i,j},  \label{eq:PEA-circuit}
\end{align}
we probabilistically apply additional noise channels $\NC^{(\delta g_{i,j})}_{i,j}$ with $\delta g_{i,j} = (g-1)/(1-2\p_{i,j})$, before each layer of the noisy circuit.  
We may also perform the unitary conjugation of the noise channels through the entire circuit to obtain a convenient representation of the noisy circuit 
\begin{align}
  \NU_g = \left(\prod_{i = 1}^{d}\prod_{j = 1}^{J}\tilde{\NC}_{i,j}^{(g)}\right) \UC_d\cdot\cdot\cdot \UC_1,  \label{eq:noisy-Ug-conjugated}
\end{align}
where $\tilde{\NC}_{i,j}^{(g)}(\rho) = (1-g\p_{i,j})\rho + g\p_{i,j} \tilde P_{i,j} \rho \tilde P_{i,j}^\dag $ and $\tilde P_{i,j} = \mathcal U_d\dots\mathcal U_{i+1}\mathcal U_i(P_{i,j})$ are Pauli strings conjugated by appropriate ideal circuit layers.

\begin{figure}[t]
\includegraphics[width=0.45\textwidth]{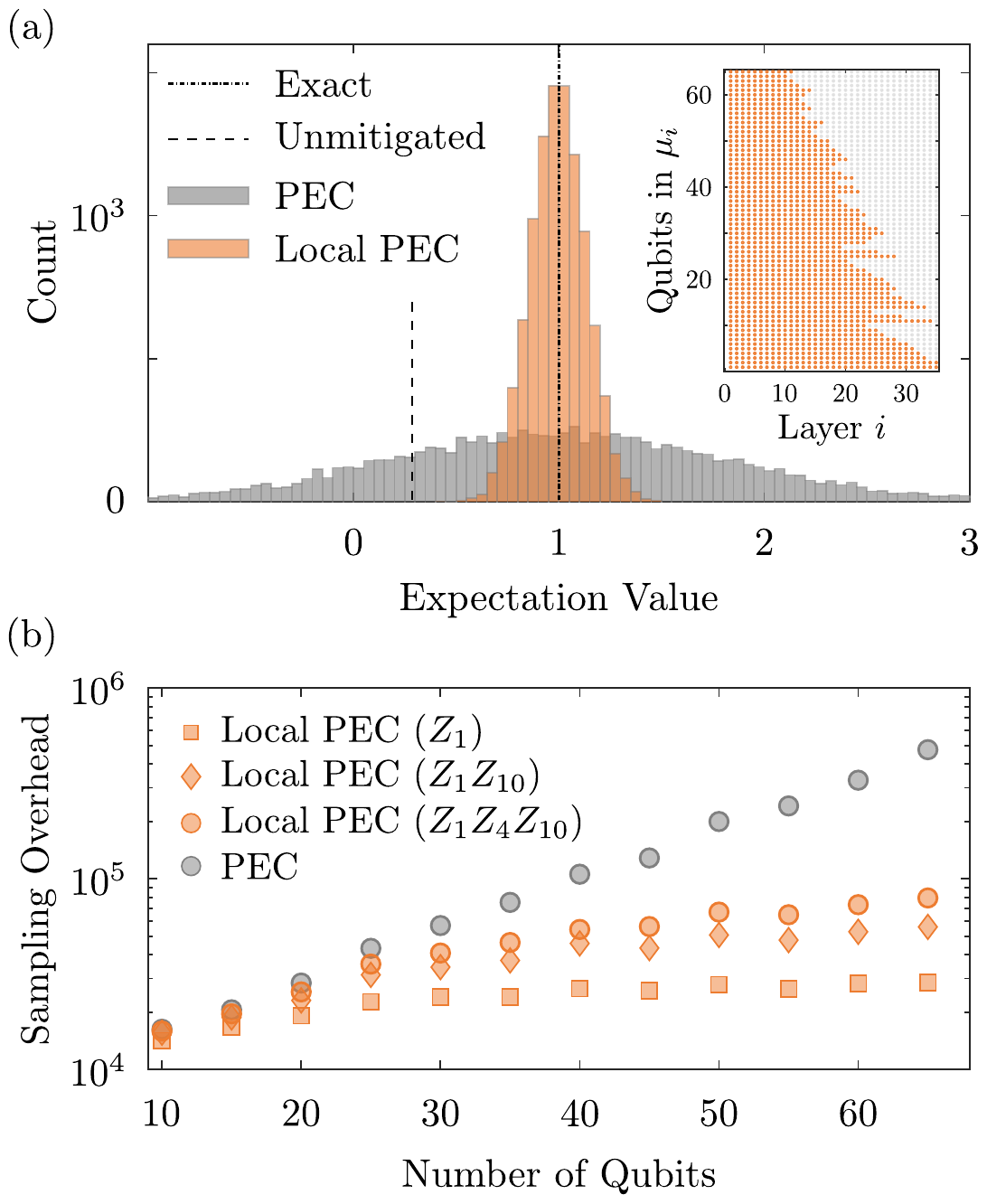}
\caption{Comparison between the standard PEC estimator and our local PEC estimator. (a) A histogram of expectation values obtained from $10^4$ independent sets, each of $10^6$ samples, using the standard PEC estimator (gray bars) and the local PEC estimator (orange bars).  
The inset indicates which qubits are in the light cone $\mu_i$ after layer $i$ of the circuit.
(b) The sampling overhead to achieve a fixed precision $\epsilon = 0.01$ using the standard estimator (gray) and the local estimators (orange).
The local estimators depend on the light cones of the observables, namely $Z_1$ (squares), $Z_1Z_{10}$ (diamonds), and $Z_1Z_4Z_{10}$ (circles). 
}
\label{fig:pec-numerics}
\end{figure}

With a gain factor $g$, we obtain the corresponding expectation value~$f(g) = \Tr(\NU_g^\dag(O)\rho)$, which can be represented as a polynomial function of $g$
\begin{align}
f(g) = \sum_{k=0}^K a_k g^k + R_{K+1}(g)~,
\end{align}
where $a_k$ are coefficients independent of $g$, $K\geq 1$ is a constant,  and $R_{K+1}$ is a remainder that sums over terms of higher orders in $g$.
Note that $f(1)$ is the expectation value obtained from the bare noisy circuit and $f(0) = a_0$ is the ideal value.
From the values at $g_0 = 1<g_1<\dots<g_K$, one constructs a linear combination~\cite{m:temmeErrorMitigationShortdepth2017a}
\begin{align}
f_{\text{ZNE}} = \sum_{\ell=0}^K \beta_\ell f(g_\ell) = f(0) + \underbrace{ \sum_{\ell = 0}^K \beta_\ell R_{K+1}(g_\ell)}_{\equiv \bar R_{K+1}}\label{eq:linear-com} 
\end{align}
by demanding that $\sum_{\ell=0}^K \beta_\ell = 1$ and $\sum_{\ell=0}^K\beta_\ell g_\ell ^k = 0$ for $k=1,\dots,K$. Solving these constraints returns the coefficients $\beta_\ell$~\cite{m:temmeErrorMitigationShortdepth2017a}:
    $\beta_\ell = \Pi_{m\neq \ell} [g_m/(g_\ell-g_m)]$.
The linear combination in \cref{eq:linear-com} gives the desired value of $f(0)$ up to a remainder that depends on high-order terms in $g$.
Since these terms correspond to multiple-error events, they should be small as long as Pauli error probabilities $\p_{i,j}$ are sufficiently small.

To bound the bias in ZNE, we use \cref{eq:noisy-Ug-conjugated} to expand $\Tr(\rho \NU_{g_\ell}^\dag(O))$ in a polynomial of $g_\ell$ and obtain an explicit form of the remainder:
\begin{align}
	R_{K+1}(g_\ell) = \sum_{k=K+1}^{dJ} g_\ell^k  
	\sum_{\nu_k} \p_{\nu_k} G_{\nu_k},\label{eq:R_K+1}
\end{align}
where $\nu_k = \{(i_1,j_1),\dots,(i_k,j_k)\}$ is a set of $k$ distinct pairs $(i,j)$, the second sum is over all possible sets $\nu_k$, $\p_{\nu_k}\equiv \p_{i_1,j_1}\dots \p_{i_k,j_k}$, 
\begin{align}
G_k \equiv \sum_{\sigma\in \{0,1\}^{k}}(-1)^{k-\Ham(\sigma)}\Tr(\UC(\rho) \tilde P_{\nu_k}^{\sigma\dag} O \tilde P_{\nu_k}^\sigma),\label{eq:Gk-def}
\end{align}
the sum in \cref{eq:Gk-def} is over all $k$-bit string $\sigma$, $\Ham(\sigma)$ is the Hamming weight of $\sigma$, and $\tilde P_{\nu_k}^\sigma \equiv \tilde P_{i_1,j_1}^{\sigma_1}\dots \tilde P_{i_k,j_k}^{\sigma_k}$ is a product of $\Ham(\sigma)\leq k$ Kraus operators.

\begin{figure}[t]
\includegraphics[width=0.42\textwidth]{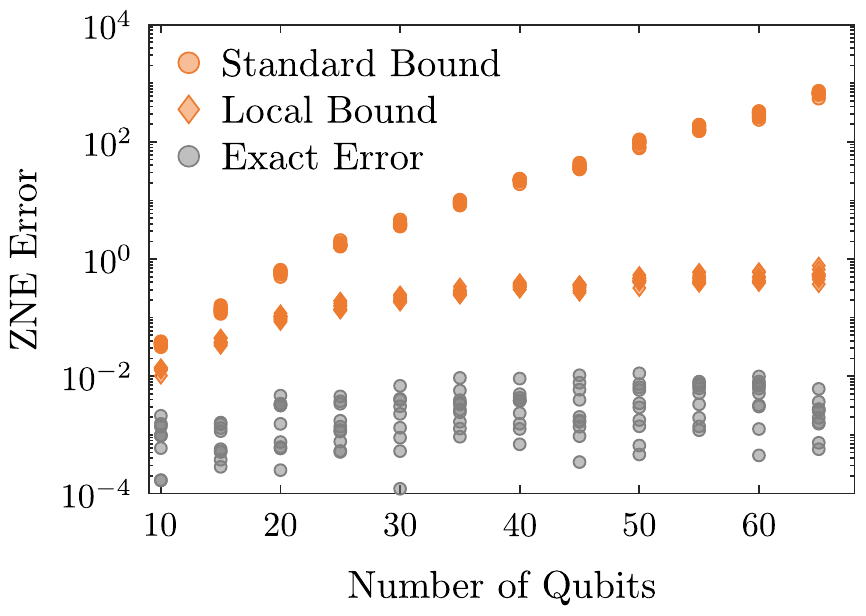}
\caption{A comparison between the exact ZNE error (gray circles), the standard bound [\cref{eq:upper-bound-1}, orange circles], and the local bound [\cref{eq:local-ZNE-bound}, orange diamonds] for measuring $Z_1$ in different randomly generated circuits.
The circuits have a fixed depth of 10 and different system sizes.
We applied ZNE with $K = 2$, $g_0 = 1,g_1 = 2,$ and $g_2 = 4$.
We estimated the exact error using $m = 10^6$ samples, making the shot noise negligible compared to the remainder.
}
\label{fig:zne}
\end{figure}

Since $\abs{G_k}\leq 2^k$, we arrive at an upper bound for the remainder in \cref{eq:linear-com}:
\begin{align} 
	\abs{\bar R_{K+1}}
	\leq \sum_{k = K+1}^{dJ} 2^k \abs{\sum_{\ell = 0}^K \beta_\ell g_\ell^k} \sum_{\nu_k} \p_{\nu_k}
	. \label{eq:upper-bound-1}
\end{align}
In practice, \cref{eq:upper-bound-1} provides a readily computable error bound for ZNE (see \cref{sec:zne-bound} for more details).
To understand how this error bound scales with the error rates, we let $\Gamma_k\equiv \abs{\sum_{\ell = 0}^K \beta_\ell g_\ell^k}$, $p \equiv \max_{i,j}\p_{i,j}$, $g = \max_{\ell} g_\ell$ and 
 obtain a more compact bound from \cref{eq:upper-bound-1}:
\begin{align} 
	 \abs{\bar R_{K+1}}
	\leq \sum_{k = K+1}^{dJ}  \Gamma_k \binom{dJ}{k} (2p)^k. \label{eq:upper-bound-2}
\end{align}
When the total experimental error rate is small, i.e. $dJgp\ll 1$, \cref{eq:upper-bound-2} implies the ZNE error is suppressed to the $K$th order: $\abs{\bar R_{K+1}} = \O{(2dJgp)^{K+1}}$.

Both \cref{eq:upper-bound-1,eq:upper-bound-2} depend on the total error rate on the entire circuit. 
But intuitively, errors that occur outside the light cone should not contribute to the remainder in \cref{eq:linear-com}.
In fact, we can realize this intuition in the error analysis by simply replacing $\NC_{i,j}$ with a trivial channel if it is outside the light cone.  
Let $\tilde \p_{i,j} = \p_{i,j}$ if $(i,j)\in \mu$, where $\mu$ is defined after \cref{eq:light-cone-def}, and $\tilde \p_{i,j} = 0$ otherwise.
\cref{eq:upper-bound-1} can be replaced by a bound that takes into account the light cone of the observable:
\begin{align} 
	\abs{\bar R_{K+1}}
	\leq \sum_{k = K+1}^{dJ} 2^k \abs{\sum_{\ell = 0}^K \beta_\ell g_\ell^k} \sum_{\nu_k} \tilde \p_{\nu_k}. \label{eq:local-ZNE-bound} 
\end{align}
We note that, in the above bounds, we ignore shot noise and assume $f(g_\ell)$ can be estimated exactly. 
In practice, we have a finite number of shots and use an estimator $\hat f(g)$ to estimate the expectation value at each gain factor $g$.
The statistical error in $f_{\text{ZNE}}$  is upper bounded by
\begin{align} 
 	 \sqrt{\sum_{\ell =0}^K \frac{1}{m_\ell}\beta_\ell^2 \Var[\hat f(g_\ell)]}
 	 \leq \sqrt{\sum_{\ell =0}^K \frac{\beta_\ell^2}{m_\ell}},
\end{align} 
where $m_\ell$ is the number of shots used to estimate $f(g_\ell)$.
Combining with \cref{eq:local-ZNE-bound}, the ZNE error is upper bounded by
\begin{align} 
	\sqrt{\sum_{\ell =0}^K \frac{\beta_\ell^2}{m_\ell}}
	+\sum_{k = K+1}^{dJ} 2^k \abs{\sum_{\ell = 0}^K \beta_\ell g_\ell^k} \sum_{\nu_k} \tilde \p_{\nu_k}. 
\end{align}

In \cref{fig:zne}, we compare the standard error bound in \cref{eq:upper-bound-1}, the local error bound in \cref{eq:local-ZNE-bound}, and the exact error after applying ZNE with $K = 2$ and $g_0 = 1, g_1 = 2, g_2 = 4$.
We randomly generate several different circuits as outlined earlier and compute the error for each of them.
We use $m = 10^6$ samples for each expectation value so that the shot noise is negligible compared to the ZNE remainder.
The circuits have a fixed CNOT-depth of 10 and a variable number of qubits $n$, from 10 to 65. 
Our local bound captures the behavior of the ZNE error better than the standard bound.
In particular, the local bound reaches a plateau as we increase $n$, whereas the standard bound grows with $n$ and becomes meaningless at large $n$.
However, the local bound still appears to be loose by about two orders of magnitude, owing primarily to the fact that we use the worst-case bound $\abs{G_k} \leq 2^k$.
In practice, the terms in \cref{eq:Gk-def} typically do not add up constructively, resulting in $G_k$ being much smaller than~$2^k$.

\medskip

\medskip 

\section{Conclusions}In this paper, we analyzed and improved error-mitigation techniques for measuring expectation values of local observables. 
Our results directly improve the performance of near-term algorithms, which heavily rely on substantial mitigation of errors in quantum devices. 
In particular, one can use our results and similar light-cone arguments to combine error mitigation with other techniques, including circuit knitting and classical shadow tomography.
\begin{acknowledgments}
We thank Andrew Eddins, Abhinav Kandala, Youngseok Kim, Antonio Mezzacapo, Alireza Seif, and Derek Wang for helpful discussions.
\end{acknowledgments}

\bibliography{zotero-generated,custom-bib}

\onecolumngrid

\pagebreak

\appendix

\vspace{0.5in}

\setcounter{theorem}{0}

\begin{center}
	{\large \textbf{Appendix for ``Locality and Error Mitigation of Quantum Circuits''} }
\end{center}

In this Appendix, we provide additional mathematical details on the local PEC estimator and on the ZNE error bound.
In \cref{sec:lopec}, we derive the mean of the local PEC estimator [\cref{eq:LoPEC-convergence}].
In \cref{sec:zne-bound}, we discuss the classical computation of the upper bound on the ZNE remainder.
In \cref{sec:numerics-specifications}, we provide detailed specifications for the numerics presented in the main text. 

\section{Mean and variance of the local PEC estimator} \label{sec:lopec}

In this section, we prove the claim in \cref{eq:LoPEC-convergence} that the mean of $\hat o_z^{\LoPEC}$ is the ideal expectation value of the observable. 
We start with the definition:
\begin{align} 
	 \mathbb E_{\bm \sigma}\mathbb E_{z}[\hat o_z^{\LoPEC}]
	 &= \sum_{\bm \sigma}\sum_{z}\prod_{(i,j)\in \mu} \gamma_{i,j} \sigma_{i,j} o_z p(\sigma_{i,j}) p(z)  \\
	 &= \sum_{\bm \sigma}\sum_{z}\prod_{(i,j)\in \mu} \gamma_{i,j} \sigma_{i,j} o_z p(\sigma_{i,j}) \Tr(\prod_{i' = 1}^d  \NU_{i'} \prod_{j' = 1}^{J}\mathcal P_{i',j'}^{\sigma_{i',j'}}(\rho)\ketbra{z} )  \\
	 &= \sum_{\bm \sigma}\prod_{(i,j)\in \mu} \gamma_{i,j} \sigma_{i,j}  p(\sigma_{i,j}) \Tr(\prod_{i' = 1}^d  \NU_{i'} \prod_{j' = 1}^{J}\mathcal P_{i',j'}^{\sigma_{i',j'}}(\rho)O )\\
	 &= \sum_{\bm \sigma}\prod_{(i,j)\in \mu} \gamma_{i,j} \sigma_{i,j}  p(\sigma_{i,j}) \Tr(\rho \prod_{i' = d}^1  \prod_{j' = 1}^{J}\mathcal P_{i',j'}^{\sigma_{i',j'}}\NU_{i'}^\dag (O) ) , 
\end{align}
where $p(\sigma_{i,j}) = 1-\p_{i,j}$ if $\sigma_{i,j} = 1$ and $p(\sigma_{i,j}) = \p_{i,j}$ otherwise.
By construction, we have
\begin{align} 
	&\sum_{\sigma_{i,j} = \pm 1}\gamma_{i,j} \sigma_{i,j} p(\sigma_{i,j}) \mathcal P_{i,j}^{\sigma_{i,j}} = \NC_{i,j}^{-1},\\
	&\sum_{\sigma_{i,j} = \pm 1} p(\sigma_{i,j}) \mathcal P_{i,j}^{\sigma_{i,j}} = \NC_{i,j}.
\end{align}
Therefore
\begin{align} 
	 \mathbb E_{\bm \sigma}\mathbb E_{z}[\hat o_z^{\LoPEC}]
	 &= \Tr(\rho \prod_{i = d}^1  \left(\prod_{j:(i,j)\in \mu} \NC_{i,j}^{-1}\right)
	 \left(\prod_{j:(i,k)\not\in \mu} \NC_{i,j}\right)
	 \prod_{j = 1}^{J}\NC_{i,j}\UC_{i}^\dag (O) ).  
\end{align}
For the pairs $(i,j)$ inside the light cone, $\NC_{i,j}^{-1}$ cancels $\NC_{i,j}$:
\begin{align} 
	 \mathbb E_{\bm \sigma}\mathbb E_{z}[\hat o_z^{\LoPEC}]
	 &= \Tr(\rho \prod_{i = d}^1  
	 \left(\prod_{j:(i,j)\not\in V} \NC_{i,j}^2\right)
	 \UC_{i}^\dag (O) ).  
\end{align}
Since we are left with only noise channels outside the light cone, we can replace them with the identity channel to obtain 
\begin{align} 
	 \mathbb E_{\bm \sigma}\mathbb E_{z}[\hat o_z^{\LoPEC}]
	 &= \Tr(\rho \prod_{i = d}^1  
	 \UC_{i}^\dag (O) ) 
	 = \avg{\UC^\dag(O)}.
\end{align}
This completes the proof of \cref{eq:LoPEC-convergent}.

\section{Computing the bounds on the ZNE remainder}\label{sec:zne-bound}
In this section, we discuss some practical details on classically computing the bounds in \cref{eq:upper-bound-1} and, similarly, \cref{eq:local-ZNE-bound}.
For large values of $dJ$, looping over $k = K+1,\dots,dJ$ and computing $\binom{dJ}{k}$ values of $\p_{\nu_k}$ appear to be inefficient.
Instead, we compute the bound using the following lemma:
\begin{lemma}\label{lem:sign-of-sum}
For all $k \geq K+1$,
\begin{align} 
		\text{sign}\left(\sum_{\ell = 0}^{K} \beta_\ell g_\ell^k\right) = (-1)^{K}. 
	\end{align}	
\end{lemma}
\begin{proof}
Let $h(x) \equiv \sum_{\ell = 0}^{K} \beta_\ell g_\ell^x$ be a continuous function of $x\in \mathbb R$. 
The function $h(x)$ is an exponential polynomial that has at most $K$ zeros.
Recall that $h(1) = h(2) = \dots = h(K) = 0$, so its $K$ zeros are exactly at $x = 1,\dots,K$.
Therefore, $h(x)$ never changes the sign for all $x > K$.
Since $h(0) = 1$ is positive and $h(x)$ changes its sign whenever $x$ increases past a zero, it follows that $\text{sign}(h(x)) = (-1)^{K}$ for all $x > K$.
\end{proof}

Using \cref{lem:sign-of-sum}, we can remove the absolute value in \cref{eq:upper-bound-1} to obtain
\begin{align} 
	\abs{\bar R_{K+1}}
	&\leq (-1)^{K+1}\sum_{k = K+1}^{dJ} 2^k \sum_{\ell = 0}^K \beta_\ell g_\ell^k \sum_{\nu_k} \p_{\nu_k} \nonumber\\
	&= (-1)^{K+1}\sum_{\ell = 0}^{K}\beta_\ell \left[\prod_{i = 1}^{d}\prod_{j = 1}^{J}(1 + 2g_\ell \p_{i,j})- \sum_{k = 0}^{K} 2^k  g_\ell^k \sum_{\nu_k} \p_{\nu_k} \right].\label{eq:rewrite-sum}
\end{align}
In the last line, we add and subtract the sum over $k = 0,\dots, K$ and use the binomial expansion to rewrite the sum over $k = 0,\dots,dJ$.
For a small $K$, both terms in \cref{eq:rewrite-sum} can be computed efficiently.

\section{Specifications for the numerical examples}\label{sec:numerics-specifications}
In this section, we provide detailed specifications for the numerics presented in the main text.
In all of our examples, we assume a system of qubits with the same connectivities as that of IBM's Ithaca device, which hosts 65 qubits in a heavy hex lattice~\cref{fig:ithaca}(a).
In examples where the number of qubits $n$ are fewer than 65, we use the first $n$ qubits according to the labels in \cref{fig:ithaca}(a).

In our examples, we benchmark the error mitigation strategies on circuits composing of layers of CNOTs and random single-qubit Clifford gates~\cref{fig:ithaca}(b).
The CNOT layers alternate between the three presets described in \cref{fig:ithaca}(b).
Every three CNOT layers is followed by a layer consisting of random single-qubit gates uniformly and independently chosen from $\{H,S\}$, where
\begin{align} 
	H = \frac{1}{\sqrt{2}}\begin{pmatrix}
		1 & 1\\
		1 & -1
	\end{pmatrix},
	\quad 
	S = \begin{pmatrix}
		1 & 0\\
		0 & i
	\end{pmatrix}.
\end{align}

Additionally, each CNOT layer is preceded by $J = 10n$ randomly generated noise channels
\begin{align} 
	 \NC_{i,j}(\rho) = e^{\mathcal L_{i,j}} \rho = (1-\p_{i,j})\rho + \p_{i,j} P_{i,j} \rho P_{i,j}^\dag.
\end{align}
For a given CNOT layer $i$, we generate a random sparse matrix $C_i \in \{0,1,2,3\}^{J\times n}$ so that its $j$th row specifies a Pauli string $P_{i,j}$, with the convention that $0\rightarrow \mathbb I, 1\rightarrow X, 2\rightarrow Y,$ and $3\rightarrow Z$.
We choose the density of $C_i$---the fraction of nonzero entries in $C_i$---so that the average weight of $P_{i,j}$ is two. 
We then choose each $\p_{i,j}$ uniformly random between $0$ and $8\times 10^{-4}$ so that the average error rate per CNOT gate is $4\times 10^{-3}$ (there are, on average, 10 error channels per CNOT gate).

\begin{figure}
\includegraphics[width=0.95\textwidth]{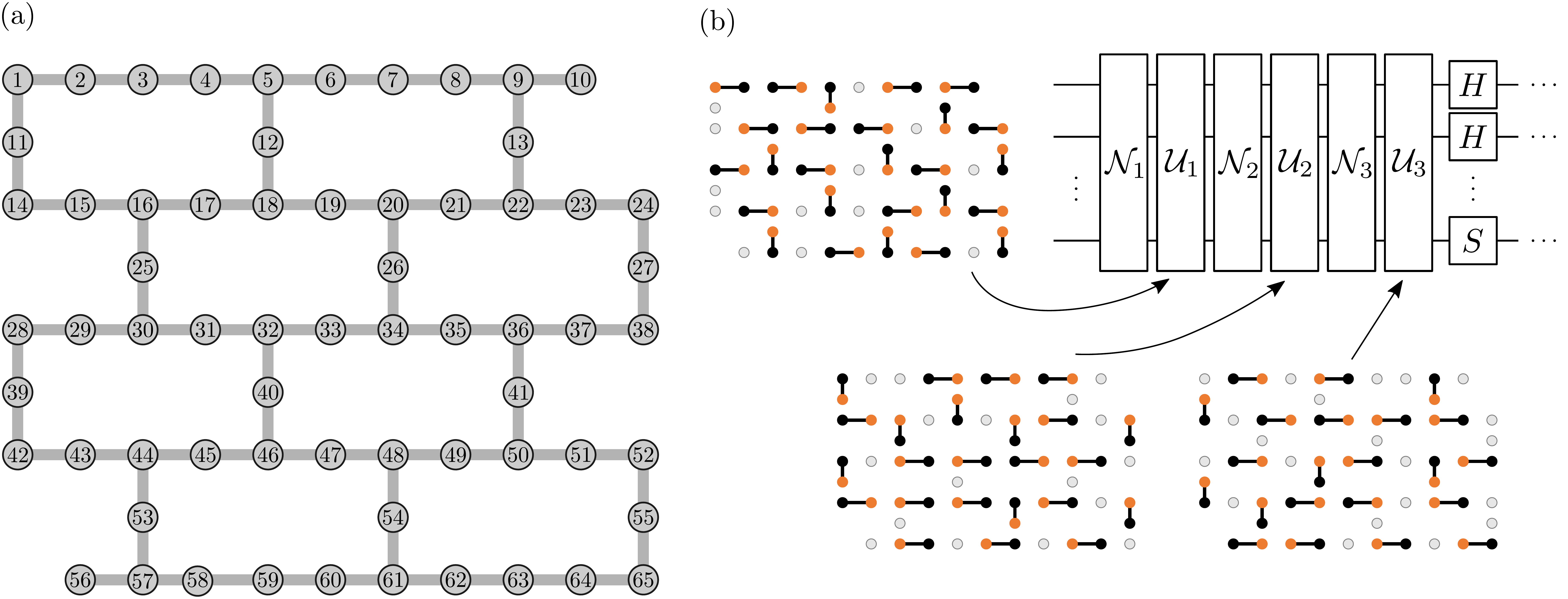}
\caption{An illustration of the systems and the circuits considered in our examples. (a) The connectivities (gray lines) between 65 qubits (gray circles) of IBM's Ithaca device. In our examples, systems of $n< 65$ qubits will consist of the first $n$ qubits ordered by the labels indicated above. 
(b) The circuits in our examples consist of CNOT layers and single-qubit gates. 
The CNOT layers alternate between three depicted presets, where a black line indicates a CNOT gate with the black qubit being the control and the orange qubit being the target. 
The single-qubit gates are chosen randomly between $H$ and $S$.}
\label{fig:ithaca}
\end{figure}

\end{document}